%% file: ms.tex
\documentclass[runningheads]{llncs}
\usepackage{amssymb}
\usepackage{amsmath}
\usepackage{comment}
\usepackage[disable]{todonotes}
\usepackage{inputenc}
\usepackage{shuffle}
\usepackage{multirow}
\usepackage{listings}
\usepackage{mathabx}
\usepackage{hyperref}

\usepackage{shuffle}

\usepackage{tikz}
\usetikzlibrary{positioning,shadows,arrows}
\usetikzlibrary{arrows,automata,positioning}
\usetikzlibrary{shapes,shapes.geometric,arrows,fit,calc,positioning,automata,}

\usepackage{diagbox}
\usepackage{slashbox}
\usepackage{tikz}
\usetikzlibrary{matrix}

\newcommand{\Syn}{\operatorname{Syn}}

\usepackage{apxproof}
\theoremstyle{plain}
\newtheoremrep{theorem}{Theorem}%[section]
\newtheoremrep{proposition}[theorem]{Proposition}
\newtheoremrep{lemma}[theorem]{Lemma}
\newtheoremrep{claim}[theorem]{Claim}
\newtheoremrep{conjecture}[theorem]{Conjecture}
\newtheoremrep{corollary}[theorem]{Corollary}
\newtheoremrep{definition}[theorem]{Definition}

\DeclareFontFamily{U}{bigshuffle}{}
\DeclareFontShape{U}{bigshuffle}{m}{n}{
  <5-8> s*[1.7] shuffle7
  <8->  s*[1.7] shuffle10
}{}
\DeclareSymbolFont{BigShuffle}{U}{bigshuffle}{m}{n}
\DeclareMathSymbol\bigshuffle{\mathop}{BigShuffle}{"001}
\DeclareMathSymbol\bigcshuffle{\mathop}{BigShuffle}{"002}

\newcommand{\NP}{\textsf{NP}}
\newcommand{\PSPACE}{\textsf{PSPACE}}

\begin{document}
%
%\title{State Complexity of Projected Languages Recognized by Permutation Automata and Commuting Letters}
%\title{State Complexity of Projected Languages of Permutation Automata}
\title{Sync-Maximal Permutation Groups Equal Primitive Permutation Groups}
%\titlerunning{State Complexity of Projection on Permutation Automata}

\author{Stefan Hoffmann\orcidID{0000-0002-7866-075X}}
\authorrunning{S. Hoffmann}
\institute{Informatikwissenschaften, FB IV, 
  Universit\"at Trier,  Universitätsring 15, 54296~Trier, Germany, 
  \email{hoffmanns@informatik.uni-trier.de}}
\maketitle              % typeset the header of the contribution
\begin{abstract}
 %The \v{C}ern\'y conjecture states a quadratic bound on the length of a shortest synchronizing word.
 %Hence, investigating so called slowly synchronizing automata such that a shortest synchronizing word has quadratic 
 %length. 
 %whose shortest synchronizing word has
 The set of synchronizing words of a given $n$-state automaton 
 forms a regular language recognizable by an automaton with $2^n - n$
 states. The size of a recognizing automaton for the set of synchronizing
 words is linked to computational problems related to synchronization
 and to the length of synchronizing words. Hence, it is natural
 to investigate synchronizing automata extremal with this property, i.e., 
 such that the minimal deterministic automaton for the set of synchronizing words
 has $2^n - n$ states.
 The sync-maximal permutation groups have been introduced in [{\sc S. Hoffmann}, Completely Reachable Automata, Primitive Groups and the State Complexity of the Set of Synchronizing Words, LATA 2021]
 by stipulating that an associated automaton to the group and a non-permutation
 has this extremal property.
 The definition is in analogy with the synchronizing groups and analog to a characterization
 of primitivity obtained in the mentioned work. The precise relation to other classes of groups
 was mentioned as an open problem. Here, we solve this open problem by showing that
 the sync-maximal groups are precisely the primitive groups. 
 Our result gives a new characterization of the primitive groups.
% Furthermore, with the new characterizations of primitivity,
% we give a new proof of a classical theorem of Jordan.
 Lastly, we explore an alternative and stronger definition than sync-maximality.
 %but the results
 %show that only the alternating and the symmetric group fulfill these conditions.

\end{abstract}

\section{Introduction}
%
% noch beziehungen zu slowly sync + ein bild von cerny automat mit 

\input{introduction}

\section{Preliminaries}
\label{sec:introduction}

 \todo{any entfernen!}

\input{preliminaries}

% eine richtung wurde ja schon gezeigt
%\section{Proof that Primitive Groups Equal Sync-Maximal Groups}
\section{Sync-Maximal Permutation Groups Equal Primitive Permutation Groups -- The Proof}

\input{main_result}

% \section{New Proof of an old Theorem by Camille Jordan}

% \input{jordan}

\section{Strongly Sync-Maximal Permutation Groups}

\input{strongly_sync_max}

{\smallskip \noindent \footnotesize
\textbf{Acknowledgement.} I thank the anonymous referees for carefully reading through
the present work and pointing out unclear formulations and typos. I am also grateful to the referee that pointed me to additional work I was not aware of. Unfortunately, due to space, I could not discuss all of them, in particular the connections to decoders and probabilistic investigations on the length of synchronizing words. The results of this submission will be incorporated into the extended version of~\cite{lata1Hoffmann21}.
}

\bibliographystyle{splncs04}
\bibliography{short}%ms,ms2,cerny,slowly_sync} 
%}
\end{document}

%% file: introduction.tex
An automaton is synchronizing if it admits a word that drives every state into a single definite state.
Synchronizing automata
have a range of applications in software testing, circuit synthesis, communication engineering and the like, see~\cite{Kohavi70,San2005,Vol2008}. The \v{C}ern\'y conjecture states that the length
of a shortest synchronizing word for a deterministic complete automaton with $n$ states
has length at most $(n - 1)^2$~\cite{Cerny64,CernyPirickaRosenauerova71}.
The best bound up to now is cubic~\cite{shitov2019}. This conjecture is one of the most famous open
problems in combinatorial automata theory~\cite{Vol2008}.
More specifically, the following bounds have been established:
\begin{center}
\begin{tabular}{ l@{\hskip 0.2in} l@{\hskip 0.2in} l }
    $2^n - n - 1$ & & (1964, \v{C}ern\'y~\cite{Cerny64}) \\ [2pt]
    $\frac{1}{2}n^3 - \frac{3}{2} n^2 + n + 1$ & & (1966, Starke~\cite{Starke66}) \\   [2pt]
    $\frac{1}{2} n^3 - n^2 + \frac{n}{2}$ & & (1970, Kohavi~\cite{Kohavi70}) \\  [2pt]
    $\frac{1}{3} n^3 - n^2 - \frac{1}{3} n + 6$ & & (1970, Kfourny~\cite{Kfoury70}) \\ [2pt]
    $\frac{1}{3} n^3 - \frac{3}{2} n^2 + \frac{25}{6} n - 4$ & & (1971, \v{C}ern\'y et al.~\cite{CernyPirickaRosenauerova71}) \\ [2pt]
    $\frac{7}{27} n^3 - \frac{17}{18} n^2 + \frac{17}{6} n - 3$ & $n \equiv 0 \pmod{3}$ & (1977, Pin~\cite{Pin77a}) \\ [2pt]
    $\left(\frac{1}{2} + \frac{\pi}{36}\right) n^3 + o(n^3)$ & & (1981, Pin~\cite{Pin81a}) \\ [2pt]
    $\frac{1}{6}n^3 - \frac{1}{6}n - 1$ & & (1983, Pin/Frankl~\cite{Frankl82,Pin83a}) \\ [2pt]
 $\alpha n^3 + o(n^3)$ & $\alpha \approx 0.1664$  & (2018, Szyku\l a~\cite{szykula2018}) \\ [2pt]
 $\alpha n^3 + o(n^3)$ & $\alpha \le 0.1654$ & (2019, Shitov~\cite{shitov2019})    
\end{tabular}
\end{center}
Furthermore, the \v{C}ern\'y conjecture~\cite{Cerny64} has been confirmed for a variety
of classes of automata, just to name a few (without further explanation): circular automata~\cite{Dubuc96,Dubuc98,Pin78a}, oriented or (generalized) monotonic automata~\cite{AnanichevVolkov03,AnanichevV05,Epp90},
automata with a sink state~\cite{Rys96}, solvable and commutative automata~\cite{Fernau019,Rys96,Rystsov97},
Eulerian automata~\cite{JKari01b}, automata preserving a chain of partial orders~\cite{Volkov09},
automata whose transition monoid contains a $\operatorname{\mathbb{Q} I}$-group~\cite{ArCamSt2017,ArnoldS06},
certain one-cluster automata~\cite{Steinberg11a}, %~\cite{BealBP11,Steinberg11a},
automata that cannot recognize $\{a,b\}^*ab\{a,b\}^*$~\cite{AlmeidaS09}, aperiodic automata~\cite{Trahtman07}, certain aperiodically $1$-contracting automata~\cite{Don16}
and automata having letters of a certain rank~\cite{BerlinkovS16}.
%, automata with simple idempotents~\cite{Rystsov2000}, automata admitting a regular collection of words~\cite{Rystsov1995b,Rystsov1995a}

\v{C}ern\'y~\cite{Cerny64} gave an infinite family of synchronizing $n$-state automata with shortest synchronizing
words of length $(n-1)^2$. Families of
synchronizing automata with shortest synchronizing words
close to $(n-1)^2$ are called \emph{slowly synchronizing}.
There are only a few families of slowly synchronizing automata known, 
see~\cite{Vol2008}.
%we only list a collection of relevant literature: ~\cite{AnanichevGV10,AnanichevV19,Don2020,DzygaFGS17,Gusev13,GusevP15,KisielewiczS15,Martugin08,Maslennikova14,Maslennikova19,Pribavkina2008,SzykulaV16}.
%with shortest synchronizing words having length $(n-1)^2$, the first
%example given by \v{C}ern\'y~\cite{Cer64}.%, or even quadratic length~\cite{}

%In fact, on average, a shortest synchronizing word has length $2n$~\cite{Higgins88}.

The set of synchronizing words of an $n$-state automaton is a regular language and can
be recognized by an automaton of size $2^n - n$~\cite{Cerny64,Maslennikova14,Starke66}. A property shared by most 
families of slowly synchronizing automata is that for them, every automaton for the set of synchronizing
words needs exponentially many states~\cite{DBLP:conf/lata/Hoffmann21a,Maslennikova14,Maslennikova19}.
Note that, of course, by taking an automaton and adjoining a letter mapping
every state to a single state, as the extremal property of the set of synchronizing words is preserved
by adding letters, automata whose sets of synchronizing words have exponential state complexity in the number of states are not necessarily slowly synchronizing. However, the evidence supports the conjecture
that slowly synchronizing automata have this extremal property.

Testing if an automaton is synchronizing is doable in polynomial time~\cite{Cerny64,Vol2008}.
However, computing a shortest synchronizing word is hard, more precisely,
the decision variant of this problem is \NP-complete~\cite{Epp90,Rys80}, even for automata over a fixed binary alphabet.
Moreover, variants of the synchronization problem for partial automata, or when restricting
the set of allowed reset words, could even be \PSPACE-complete~\cite{FernauGHHVW19,Martyugin12}.

The size of a smallest automaton for the set of synchronizing words seems to be also related to the difficulty
to compute a shortest synchronizing word, or a synchronizing word subject to certain constraints.
A first result in this direction was the realization
that for commutative automata, i.e., where each permutation of an input word leads to the same 
state, and a fixed alphabet, we do not have such an exponential blowup for the size of the minimal automaton for set of synchronizing words~\cite{HoffmannCocoonExtended,Hoffmann2021constrsync}. As a consequence, the constrained
synchronization problem for commutative input automata and a fixed constraint is always solvable in polynomial time~\cite{HoffmannCocoonExtended,Hoffmann2021constrsync}. Note that for commutative input automata and a fixed alphabet,
computing a shortest synchronizing can be done in polynomial time~\cite{Martyugin09}

So, it is natural to focus on synchronizing automata such that the smallest automaton for the set of synchronizing
words has maximal possible size.

After realizing that for certain special cases for which the \v{C}ern\'y conjecture
was established~\cite{ArnoldS06,Dubuc96,Dubuc98,Pin78a,Rystsov2000}, this was due to the reason that certain
permutation groups were contained in the transformation monoid of the automaton,
the notion of synchronizing permutation groups was introduced~\cite{ArCamSt2017,ArnoldS06}.
These are permutation groups with the property that if we adjoin a non-permutation
to it, the generated transformation monoid contains a constant map. It was shown that
these groups are contained strictly between the $2$-transitive and the primitive groups~\cite{ArCamSt2017,Neumann09}.
Meanwhile, a lot of related permutation groups
have been introduced or linked to the synchronizing groups, for example: spreading, separating, $\mathbb Q$I-groups.
See~\cite{ArCamSt2017} for a good survey and definitions. Furthermore, permutation groups in general
have been investigated with respect to the properties of resulting transformation
monoid if non-permutations were added~\cite{ARAUJO2021a,araujo2019,AraujoCameron2014,ArCamSt2017}.

The completely reachable automata have been introduced by Volkov \& Bondar~\cite{DBLP:conf/dcfs/BondarV16,DBLP:conf/dlt/BondarV18}.
This is a stronger notion than being synchronizing by stipulating that, starting from the whole state set,
not only some singleton set is reachable, but every non-empty subset of states is reachable by some word.
In fact, this property was also previously observed for many classes of synchronizing
automata~\cite{Don16,DBLP:conf/lata/Hoffmann21a,Maslennikova14,Maslennikova19}.

It has been proven in~\cite{lata1Hoffmann21} that a permutation group of degree $n$ is primitive if and only
if in the transformation monoid generated by the group and
an arbitrary non-permutation with an image of size $n - 1$, there exists, for every non-empty subset of the permutation domain, 
an element mapping the whole permutation domain to this subset, or said differently that
an associated automaton is completely reachable.
In the same paper~\cite{lata1Hoffmann21} the sync-maximal permutation groups were introduced
by stipulating that, for an associated $n$-state automaton, the smallest automaton for the set of synchronizing words\todo{nur exponentiell fordern, $2^{n/2}$ geht auch? was kommt raus} has size $2^n - n$. It was shown that
the sync-maximal permutation groups are contained between the $2$-homogeneous and the primitive permutation groups,
and it was posed as an open problem if they are properly contained between them, and if so, what
the precise relation to other permutation groups is. 

Here, we solve this open problem by showing that the sync-maximal permutation groups are precisely the primitive
permutation groups, which also yields new characterizations of the primitive permutation groups.

%% file: preliminaries.tex
%\subsection{General Notions}

\todo{$|w|$ entfernt, brauche ich auch nicht?}

Let $\Sigma$ be a finite set of symbols,
 called an \emph{alphabet}. 
 By $\Sigma^{\ast}$, we denote
the \emph{set of all finite sequences}, i.e., of all words or strings. The \emph{empty word}, i.e., the 
finite sequence of length zero, is denoted by~$\varepsilon$. %We set $\Sigma^+ = \Sigma^* \setminus \{\varepsilon\}$.
%For a given word $w \in \Sigma^*$, we denote by $|w|$
%its \emph{length}. 
The subsets of $\Sigma^{\ast}$
are called \emph{languages}. 
For $n > 0$, we set $[n] = \{0,1,\ldots, n-1\}$ and $[0] = \emptyset$.
For a set $X$, we denote the \emph{power set} of $X$ by $\mathcal P(X)$, i.e,
the set of all subsets of $X$.
Every function $f : X \to Y$
induces a function $\hat f : \mathcal P(X) \to \mathcal P(Y)$
by setting $\hat f(Z) := \{ f(z) \mid z \in Z \}$.\todo{brauche ich funktionen auf mengen allgemein?}
Here, we denote this extension also by~$f$.
Let $k \ge 1$. A \emph{$k$-subset} $Y \subseteq X$ is a finite set of cardinality~$k$.
A $1$-set is also called a \emph{singleton set}.
For functions $f : A \to B$ and $g : B \to C$, the \emph{functional composition}
$gf : A \to C$ is the function $(gf)(x) = g(f(x))$, i.e., the function on the right is applied first\footnote{In group theory, usually the other convention is adopted, but we stick to the convention most often seen in formal
language theory.}. A function $f : X \to X$ is called \emph{idempotent}
if $f^2 = f$.

A \emph{semi-automaton} is a triple $\mathcal A = (\Sigma, Q, \delta)$
where $\Sigma$ is the \emph{input alphabet}, $Q$ the finite set of \emph{states}
and $\delta : Q \times \Sigma \to Q$ the \emph{transition function}.
The transition function $\delta : Q \times \Sigma \to Q$
extends to a transition function on words $\delta^{\ast} : Q \times \Sigma^{\ast} \to Q$
by setting $\delta^{\ast}(s, \varepsilon) := s$ and $\delta^{\ast}(s, wa) := \delta(\delta^{\ast}(s, w), a)$
for $s \in Q$, $a \in \Sigma$ and $w \in \Sigma^{\ast}$. In the remainder we drop
the distinction between both functions and also denote this extension by $\delta$.
For $S \subseteq Q$ and $w \in \Sigma^*$, we write $\delta(S, w) = \{ \delta(s, w) \mid s \in S \}$.
%and $\delta^{-1}(S, w) = \{ q \in Q \mid \delta(q, w) \in S \}$.
A state $q \in Q$ is \emph{reachable} from a state $p \in Q$, if
there exists $w \in \Sigma^*$ such that $\delta(p, w) = q$.
%A subset $S \subseteq Q$ is a \emph{strongly connected component}, if $S$ is maximal
%with the property that all pairs of states in $S$ are reachable from each other.
%We call $\mathcal A$ \emph{strongly connected}, if it has precisely one strongly connected component.
% The semi-automaton $\mathcal A$ is \emph{synchronizing}, if there exists a word $w \in \Sigma$
% such that $|\delta(Q, w)| = 1$. Then, $w$ is called a \emph{synchronizing word (for $\mathcal A)$}.
The \emph{set of synchronizing words} is $\Syn(\mathcal A) = \{ w \in \Sigma^* \mid |\delta(Q, w)| = 1 \}$.
A semi-automaton $\mathcal A = (\Sigma, Q, \delta)$
is called \emph{synchronizing}, if $\Syn(\mathcal A) \ne \emptyset$.
We call $\mathcal A$ \emph{completely reachable}, if for each non-empty $S \subseteq Q$
there exists a word $w \in \Sigma^*$ such that $\delta(Q, w) = S$.
Note that every completely reachable automaton is synchronizing.

A \emph{(finite) automaton} is a quintuple $\mathcal A = (\Sigma, Q, \delta, q_0, F)$
where $(\Sigma, Q, \delta)$ is a semi-automaton, $q_0 \in Q$ is the \emph{start state}
and $F \subseteq Q$ the set of \emph{final states}. The \emph{languages recognized (by $\mathcal A$)}
is $L(\mathcal A) = \{w \in \Sigma^* \mid \delta(q_0, w) \in F \}$. 
An automaton with a start state and a set of final states is used
for the description of languages, whereas, when we consider a semi-automaton,
we are only concerned with the transition structure of the automaton itself.
When the context is clear, we also call semi-automata simply automata
and concepts and notions that do not use the start state or the final state
carry over from semi-automata to automata and vice versa.

A language recognized by a finite automaton is called \emph{regular}.
An automaton $\mathcal A$ has the least number of states to recognize a language~\cite{HopUll79}, i.e., is a \emph{minimal automaton}, 
if and only if every state is reachable from the start state and
every two distinct states $p, q \in Q$ are \emph{distinguishable}, i.e., 
there exists $w \in \Sigma^*$ such that precisely one of the states $\delta(p, w)$
and $\delta(q, w)$ is a final state, but not the other. A minimal automaton
is unique up to isomorphism~\cite{HopUll79}, where two automata are isomorphic if one can be obtained from the other
by renaming of states. Hence, we can speak about the minimal automaton.

% If the context is clear, we call
% semi-automata also simply automata, and concepts defined for semi-automata
% are defined in the same way for automata.

If $\mathcal A = (\Sigma, Q, \delta)$
is a semi-automaton with a non-empty state set, then define $\mathcal P_{\mathcal A} = (\Sigma, \mathcal P(Q) \setminus\{\emptyset\}, \hat \delta, Q, F)$ where $\hat \delta : \mathcal P(Q) \times \Sigma \to \mathcal P(Q)$ is the extension $\hat \delta(S, u) = \{ \delta(s,u) \mid s \in S \}$, for $S \subseteq Q$ and $u \in \Sigma^*$,
of $\delta$ to subsets of states
and $F = \{\{q\} \mid q \in Q \}$. As for functions $f : X \to Y$ introduced above,
we drop the distinction between $\delta$ and $\hat \delta$
and denote both functions by $\delta$.
We have, %$L(\mathcal P_{\mathcal A}) = \Syn(\mathcal A)$.
$\Syn(\mathcal A) = \{ w \in \Sigma^* \mid \delta(Q, w) \in F \}$.
The states in $F$ can be merged to a single state
to get a recognizing automaton for $\Syn(\mathcal A)$.
So, $\Syn(\mathcal A)$ is recognizable by an
automaton with $2^{|Q|} - |Q|$ states.

Let $n \ge 0$. %Set $[n] = \{ 0,1, \ldots, n-1 \}$.
Denote by $\mathcal S_n$ the \emph{symmetric group on $[n]$}, i.e., the group of all permutations of $[n]$.
A \emph{permutation group (of degree $n$)}
is a subgroup of $\mathcal S_n$.
%For $n > 1$, the \emph{alternating group} is the unique subgroup of size $n! / 2$
%in $\mathcal S_n$, see~\cite{cameron_1999}.
%\todo{def orbit entfernt.}
%The \emph{orbit} of an element $i \in [n]$ for a permutation group $G$ is the set $\{ g(i) \mid g \in G\}$.
A permutation group $G$ over $[n]$ is \emph{primitive}, if it preserves no non-trivial equivalence
relation\footnote{The trivial equivalence relations on $[n]$ are $[n] \times [n]$
and $\{ (x,x) \mid x \in [n] \}$.}
on $[n]$, i.e., for no non-trivial equivalence relation $\sim \subseteq [n]\times [n]$
we have $p \sim q$ if and only if $g(p) \sim g(q)$ for all $g \in G$ and $p, q \in [n]$ (recall that the elements of $G$
are functions from $[n]$ to $[n]$).
Equivalently, a permutation group is primitive if there
does not exist a non-empty proper subset $\Delta \subseteq [n]$
with $|\Delta| > 1$ such that, for every $g \in G$,
we have $g(\Delta) = \Delta$ or $g(\Delta)\cap \Delta = \emptyset$.
A permutation group $G$ over $[n]$ is called \emph{k}-homogeneous for some $k \ge 1$,\todo{def $k$-hom entfernen?}
if for every two $k$-subsets $S,T$ of $[n]$, there exists $g \in G$
such that $g(S) = T$. A \emph{transitive} permutation group is the same
as a $1$-homogeneous permutation group.
Note that here, all permutation groups with $n \le 2$ are primitive, and for $n > 2$
every primitive group is transitive. Because of this, some authors exclude
the trivial group for $n = 2$ from being primitive.
A permutation group $G$ over $[n]$ is called \emph{$k$-transitive} for some $k \ge 1$,
if for two $k$-tuples $(p_1, \ldots, p_k), (q_1, \ldots, q_k) \in [n]^k$,
there exists $g \in G$
such that $(g(p_1),\ldots, g(p_k)) = (q_1, \ldots, g_k)$.

%\paragraph{Transformations}
By $\mathcal T_n$, we denote the set of all maps on $[n]$.
The elements of $[n]$ are also called \emph{points} in this context.
A submonoid of $\mathcal T_n$ for some $n$ is called a \emph{transformation monoid}.
If the set $U$ is a submonoid (or a subgroup) of $\mathcal T_n$ (or $\mathcal S_n$)
we denote this by $U \le \mathcal T_n$ (or $U \le \mathcal S_n$).
For a set $A \subseteq \mathcal T_n$ (or $A \subseteq \mathcal S_n$),
we denote by $\langle A \rangle$ the submonoid (or the subgroup) generated
by $A$. If $A = \{a_1, \ldots, a_m\}$ we also write $\langle a_1, \ldots, a_m \rangle = \langle A \rangle$.
Let $\mathcal A = (\Sigma, Q, \delta)$ be a semi-automaton and
for $w \in \Sigma^*$ define $\delta_w : Q \to Q$
by $\delta_w(q) = \delta(q,w)$ for all $q \in Q$.
Then, we can associate with $\mathcal A$ the \emph{transformation monoid
of the automaton} $\mathcal T_{\mathcal A} = \{ \delta_w \mid w \in \Sigma^* \}$, where we can identify $Q$
with $[n]$ for $n = |Q|$.
%Seeing the letters as transformations of the state set,
%$\mathcal T_{\mathcal A} = \langle \Sigma \rangle$.
We have $\mathcal T_{\mathcal A} = \langle \{ \delta_x \mid x \in \Sigma \} \rangle$.
The \emph{rank} of a map $f : [n] \to [n]$ is the cardinality of its image.
For a given semi-automaton $\mathcal A = (\Sigma, Q, \delta)$, the \emph{rank of a word} $w \in \Sigma^*$
is the rank of $\delta_w$. 
We call two sets $S, T \subseteq [n]$ \emph{distinguishable}
in a transformation monoid $M \le \mathcal T_n$ if there exists an element in $M$
mapping precisely one of both sets to a singleton set 
and the other to a non-singleton set.

%\paragraph{Permutation Groups} 

%\subsection{Known Results}

\todo{def state collapsen entfernt.}
%We say that a word $w \in \Sigma^*$ \emph{collapses} two states $p,q \in Q$
%of a semi-automaton $\mathcal A = (\Sigma, Q, \delta)$
%if $\delta(p,w) = \delta(q,w)$. Similarly, a map $f : [n] \to [n]$ collapsed
%two points $p, q \in [n]$ if $f(p) = f(q)$.

The following implies that we can check if a given semi-automaton is synchronizing
by only looking at pairs of states~\cite{Cerny64,Vol2008}. The proof basically
works by repeatedly collapsing pairs of states to construct a synchronizing word~\cite{Vol2008}. It implies a
polynomial time procedure to check synchronizability~\cite{Vol2008}.

\begin{theorem}[\v{C}ern\'y~\cite{Cerny64,Vol2008}]
\label{thm:pair_sync_criterion}
 Let $\mathcal A = (\Sigma, Q, \delta)$.
 Then, $\mathcal A$ is synchronizing if and only if for each $p, q \in Q$
 there exists $w \in \Sigma^*$ such that $\delta(p, w) = \delta(q, w)$.
 Hence, a transformation monoid $M \le \mathcal T_n$
 contains a constant map
 if and only if every two points can be mapped to a single point
 by elements in $M$. %sprechweise collapsable
\end{theorem}

The next result appears in~\cite{AraujoCameron2014,ArCamSt2017} and despite it was never clearly spelled out by Rystsov himself, it is implicitly present in arguments used in~\cite{Rystsov2000}.

\begin{theorem}[Rystsov~\cite{AraujoCameron2014,ArCamSt2017,Rystsov2000}]
\label{thm:rystsov} 
 A permutation group $G$ on $[n]$ 
 is primitive if and only if, for every map $f : [n] \to [n]$
 of rank $n - 1$, the transformation monoid
 $\langle G \cup \{f\} \rangle$ contains a constant map.
\end{theorem}

% gruppen primitiv, sync -> automat, gruppe primitiv <=> in gezeigt und
% sync maximal

 In~\cite{lata1Hoffmann21}, the following characterization
 of the primitive permutation groups, connecting them 
 to completely reachable automata, was shown.
 
 \begin{theorem}[Hoffmann~\cite{lata1Hoffmann21}]
 \label{thm:compl_reach}
  Let $G = \langle g_1, \ldots, g_k \rangle \le S_n$.
  Then the following are equivalent:
  \begin{enumerate}
  \item $G$ is primitive;
  
  \item for every transformation
  $f : [n] \to [n]$ of rank $n - 1$ and non-empty $S \subseteq [n]$,
  there exists $g \in \langle G \cup \{f\} \rangle$
  such that $g([n]) = S$;
  
  \item for every transformation
  $f : [n] \to [n]$ of rank $n - 1$, the semi-automaton
  $\mathcal A = (\Sigma, Q, \delta)$,
  with $\Sigma = \{g_1, \ldots, g_k,f\}$, $Q = [n]$
  and $\delta(m, g) = g(m)$ for all $m \in [n]$
  and all $g \in \Sigma$, is completely reachable.
  \end{enumerate}
 \end{theorem}
 
 In the same work~\cite{lata1Hoffmann21}, in analogy with
 Theorem~\ref{thm:rystsov}
 and ~\ref{thm:compl_reach},
 the \emph{sync-maximal permutation groups} were introduced.
 
\begin{definition}
\label{def:sync-maximal}
  A permutation group $G = \langle g_1, \ldots, g_k \rangle \le \mathcal S_n$ 
  is called \emph{sync-maximal}, if for every map $f \colon [n] \to [n]$ of rank $n - 1$,
  for the automaton $\mathcal A = (\Sigma, [n], \delta)$
  with $\Sigma = \{g_1, \ldots, g_k, f\}$ and
  $\delta(m, g) = g(m)$ for $m \in [n]$ and $g \in \Sigma$,
  the minimal automaton of $\Syn(\mathcal A)$
  has $2^n - n$ states.
  %we have $\stc(\Syn(\mathcal A)) = 2^n - n$,
\end{definition}
 
 This definition is independent of the choice of generators for $G$.
 In purely combinatorial language, using the characterization
 of the minimal automaton, applied to $\mathcal P_{\mathcal A}$, this means the following.
 
 \begin{theoremrep}
  A permutation group $G \le S_n$ is sync-maximal if and only if 
  for every transformation $f : [n] \to [n]$ of rank $n - 1$,
  we have that 
  (1) for every non-empty subset $S \subseteq [n]$ of size at least two\footnote{As $f$ has rank $n - 1$, this also implies
  that at least one singleton set is reachable. In fact, even more holds true, with~\cite[Lemma 4.1]{lata1Hoffmann21}
  we can deduce that $G$ is transitive and so for every non-empty $S \subseteq [n]$ there exists $h \in \langle G \cup \{f\}\rangle$
  with $S = h([n])$.} there exists $h \in \langle G \cup \{f\}\rangle$ such that $S = h([n])$, and
  (2) for two distinct non-empty and non-singleton subsets of $[n]$, there exists a transformation
   in $\langle G \cup \{f\} \rangle$ mapping precisely one to a singleton but not the other.
 \end{theoremrep}
 \begin{proof}
  As $G$ is finite, there exists a generating set $g_1, \ldots, g_k$.
  If we associate the semi-automaton $\mathcal A = (\Sigma, [n], \delta)$ % direkt g_i(m)
  with $\Sigma = \{g_1, \ldots, g_k, f\}$ and $\delta(m, g) = g(m)$
  for $g \in \Sigma$ and $m \in [n]$ to $G$,
  then, as $G$ is sync-maximal, the minimal automaton for $\Syn(\mathcal A)$
  has $2^n - n$ states.
  Hence, in $\mathcal P_{\mathcal A}$
  all non-empty and non-singleton sets are reachable
  and at least one singleton set, which gives the first condition,
  and every two distinct non-empty and non-singleton sets are distinguishable
  in $\mathcal P_{\mathcal A}$, which gives the second claim.

  Conversely, suppose conditions (1) and (2) hold true
  and let $f : [n] \to [n]$ be of rank $n - 1$.
  If we associate an automaton $\mathcal A$ with $G$ and $f$
  as before, then condition (2) says that
  that all states in $\mathcal P_{\mathcal A}$
  are distinguishable.
  Condition (1) implies that every subset of size at
  least two is reachable in $\mathcal P_{\mathcal A}$.
  Let $a,b \in [n]$ be such that $f(a) = f(b)$
  and $h \in \langle G \cup \{f\}\rangle$
  with $h([n]) = \{a,b\}$.
  Then, $|f(h([n]))| = 1$
  and so at least one singleton set is reachable.
  In fact, by Lemma~\cite[Lemma 4.1]{lata1Hoffmann21},
  $G$ is transitive and so every subset is reachable
  in $\mathcal P_{\mathcal A}$.\qed
 \end{proof}

 In~\cite{lata1Hoffmann21} it was shown that every sync-maximal permutation groups is primitive.
 The main result of the present work is that the converse implication holds true.

 \begin{proposition}[\cite{lata1Hoffmann21}]\label{prop:sync_max_implies_primitive}
  Every sync-maximal permutation group is primitive.
 \end{proposition}

 In~\cite[Lemma 3.1]{lata1Hoffmann21} it was shown that distinguishability
 of all sets reduces to distinguishablity of the $2$-subsets.
 Formulated without reference to automata, this gives the next result.
 
 \begin{theoremrep}[\cite{lata1Hoffmann21}]
 \label{thm:2-set_distinguishable}
 \todo{brauche $\mathcal T_n$ doch nur hier?}
  Let $M \le \mathcal T_n$ be a transformation monoid.
  Then, for every two distinct non-empty and non-singleton  $S, T \subseteq [n]$
  there exists a transformation in $M$ mapping precisely one to a singleton but not the other
  if and only if this condition holds true for every two distinct $2$-subsets of $[n]$.
 \end{theoremrep}
 \begin{proof}
  The method of proof used in~\cite[Lemma 3.1]{lata1Hoffmann21}
  works in the general case of transformation monoids. However, we give a simple self-contained (and 
  shorter than the original) proof here\footnote{This argument, for which I am thankful, was communicated to me by an anonymous reviewer of the present work.}.
  Assume the claim holds true for the $2$-subsets of~$[n]$.
  We do induction on $|S| + |T|$.
  As by assumption $\min\{|S|, |T|\}\ge 2$,
  we have $|S| + |T| \ge 4$.
  If $|S| + |T| = 4$, then both 
  are $2$-subsets and this is precisely the assumption.
  If $|S| + |T| > 4$, then, as $S \ne T$, there exist
  $\{a,b\} \subseteq S$ and $\{c,d\} \subseteq T$
  with $\{a,b\} \ne \{c,d\}$.
  By assumption there exists $g \in M$
  mapping precisely one to a singleton. Without loss of generality, suppose
  this is $\{a,b\}$. Then $|g(S)| < |S|$
  and $|g(S)| + |g(T)| < |S| + |T|$.
  Then, either $g(S)$ is a singleton set and $|g(T)| \ge 2$
  or, by the induction hypothesis, we find $h \in M$ mapping precisely one
  of $g(S)$ and $g(T)$ to a singleton set.
  In this case, $hg$ maps precisely one of $S$ and $T$ to a singleton set.\qed 
 \end{proof}

%  Then, it was shown in~\cite{lata} that the
%  sync-maximal permutation groups are contained
%  between the $2$-homogeneous and the primitive
%  permutation groups. However, the precise relation
%  to these, and also to other classes like the synchronizing, spreading, $\mathbb Q$I-groups and so on was left as an open problem.
%  Here, we solve this open problem by showing
%  that, in fact, the property of being sync-maximal
%  gives another characterization of the primitive permutation groups.
%  So, the main contribution of the present work will be the next statement.
 
%\begin{toappendix}
The next lemma is obvious, as it basically states the definition of injectivity
for the restriction of a function to a subset, and stated for reference.

\begin{lemma}
\label{lem:transformation_acts_injective}
 Let $f : [n] \to [n]$ and $S \subseteq [n]$.
 Then, $|S \cap f^{-1}(x)| \le 1$ for each $x \in [n]$
 if and only if $f$ acts injective on $S$, i.e., $|f(S)| = |S|$.
\end{lemma}
%\end{toappendix}

%% file: main_result.tex
 Here, we will prove the following theorem.
 
  \begin{theorem}
 \label{thm:sync_maximal_equals_primitive}
  Let $G = \langle g_1, \ldots, g_k \rangle \le S_n$.
  Then the following are equivalent:
  \begin{enumerate}
  \item $G$ is primitive;
  
  \item for every transformation $f : [n] \to [n]$ of rank $ n - 1$
   and $\{a,b\},\{c,d\} \subseteq [n]$
   with $\{a,b\}\ne \{c,d\}$,
   there exists $g \in \langle G \cup \{f\} \rangle$ such
   that precisely one of the subsets $g(\{a,b\})$ and $g(\{c,d\})$
   is a singleton but not the other;
   
  \item $G$ is sync-maximal. \todo{reihenfolge ändern?}
   
%   \item for every transformation $f : [n] \to [n]$
%   of rank $n - 1$, for the semi-automaton $\mathcal A = (\Sigma, [n], \delta)$
%   with $\Sigma = \{g_1, \ldots, g_k, f\}$ and $\delta(i, g) = g(i)$
%   with $i \in [n]$ and $g \in \Sigma$, the minimal automaton of $\Syn(\mathcal A)$
%   has $2^n - n$ states.
  %automaton, wie in lata. folgerung wenn primitive gruppe enthält, dann sync-maximal.
  \end{enumerate} %\todo{erwähnen, letzte folgerung hieß sync-maximal, oder extra definieren.}
 \end{theorem}

 At the heart of our result is the following statement.
 
\begin{proposition} % ne, habe doch das beispiel S_n wr S_m ...
 %Every primitive group is sync-max.
\label{prop:prim_grp_distinguishes_2_sets}
 Let $G \le S_n$ be a permutation group
 and $f : [n] \to [n]$ be an idempotent map of rank $n - 1$.
 Suppose $\langle G \cup \{f\} \rangle$ contains a constant map.
 Then, for all $\{a,b\}, \{c,d\} \subseteq [n]$
 with $\{a,b\} \ne \{c,d\}$
 there exists a transformation in $\langle G \cup \{f\} \rangle$
 mapping precisely one set to a singleton but not the other.
\end{proposition}
%\begin{proofsketch}
% Todo.
%\end{proofsketch}
\begin{proof}
 Suppose $G = \langle g_1, \ldots, g_k \rangle$.
 Let $\{a,b\}, \{c,d\} \subseteq [n]$ be two distinct $2$-sets
 and $f : [n] \to [n]$ being idempotent and of rank $n - 1$.
 Without loss of generality, we can suppose $f(0) = f(1) = 1$ and $f(i) = i$ for $i \in \{2,\ldots,n-1\}$.
 As $\langle G \cup \{f\} \rangle$ contains a constant map, 
 we can map $\{a,b\}$ to a singleton set.
 Choose a transformation $h \in \langle G \cup \{f\}\rangle$
 represented by a shortest possible word in the generators of $G$ and $f$ such that $|h(\{a,b\})| = 1$.
 Then, we can write $h = f u_m f u_{m-1} f \cdots u_2 f u_1$ with $u_i \in G$, $m \ge 1$.
% and the least number of $f$'s in it.
 Note that, by the minimal choice, the transformation
 $f$ is applied at the end\footnote{Recall that by our convention, the leftmost function is applied last.}.
 If $h(\{c,d\})$ is not a singleton set, we are done. So, suppose
 $h(\{c,d\})$ is also a singleton set.
 For $i \in \{1,\ldots,m\}$, set $h_i = f u_i f u_{i-1} f \cdots u_2 f u_1$.
 Then, $h = h_m$.
 By minimality of the representation in the generators of $G$ and $f$,
 for all $i \in \{1,\ldots,m-1\}$ we have
 $|h_i(\{a,b\})| = 2$.
 Hence, if there exists $i \in \{1,\ldots,m-1\}$
 such that $h_i(\{c,d\})$ 
 is a singleton set, we are also done.
 So, suppose this is not the case.
 
 Set $g_i = u_i f u_{i-1} f \cdots u_2 f u_1$. Then, $h_i = fg_i$.
 We must have $0 \in g_i(\{a,b\})$ for all $i$,
 for otherwise, as $f$ acts as the identity on $\{1,\ldots,n-1\}$,
 we can leave $f$ out, i.e., $g_i(\{a,b\}) = h_i(\{a,b\})$, in the expression for $h$
 and get a shorter representing word, contradicting the minimal choice
 of the expression representing $h$.
 Similarly, $0 \in g_i(\{c,d\})$ for all $i$,
 as otherwise we can leave a single instance of $f$ out again and have a word that maps $\{c,d\}$
 to a singleton, but not $\{a,b\}$ by the minimal choice of $h$ in the length of a representing
 word.
 
 Note that, as $|g_m(\{a,b\})| = |g_m(\{c,d\})| = 2$, $|h_m(\{a,b\})| = |h_m(\{c,d\})| = 1$
 and $h_m = fg_m$, we must have
 $g_m(\{a,b\}) = g_m(\{c,d\}) = \{0,1\}$.
 
 \todo{direkt minimal wählen mit beide auf singleton und dann widerspruch?}
 Next, we will show by induction on $j \in \{1,\ldots,m\}$ that $g_j(\{a,b\}) = g_j(\{c,d\})$,
 where the base case is $j = m$.
 Then, $g_1(\{a,b\}) = g_1(\{c,d\})$
 implies $\{a,b\} = \{c,d\}$
 as $g_1 = u_1 \in G$ is a permutation.
 However, this is a contradiction
 as both $2$-sets are assumed to be distinct.
 Hence, the case that $h_m$ maps both to a singleton
 and $h_i$ for all $i \ne m$ maps both to $2$-sets
 is not possible.
 As noted, for the base
 case we have
 $g_m(\{a,b\}) = g_m(\{c,d\}) = \{0,1\}$.
 Now suppose $j \in \{1,\ldots,m-1\}$
 and $g_{j+1}(\{a,b\}) = g_{j+1}(\{c,d\})$.
 Then, as $g_{j+1} = u_{j+1} f g_j = u_{j+1} h_j$,
 we can deduce $h_j(\{a,b\}) = h_j(\{c,d\})$
 as they only differ by the application of the permutation $u_{j+1} \in G$.
 As written above $0 \in g_j(\{a,b\}) \cap g_j(\{c,d\})$.
 This implies,
 as $|h_j(\{a,b\})| = |h_j(\{c,d\})| = 2$,
 that $1 \notin g_j(\{a,b\}) \cup g_j(\{c,d\})$.
 So, $|f^{-1}(x) \cap g_j(\{a,b\})| \le 1$
 and $|f^{-1}(x) \cap g_j(\{c,d\})| \le 1$ for every $x \in [n]$.
 As $h_j = fg_j$, we can write $f(g_j(\{a,b\})) = f(g_j(\{c,d\}))$.
 Applying Lemma~\ref{lem:transformation_acts_injective} then yields $g_j(\{a,b\}) = g_j(\{c,d\})$.\qed 
\begin{comment}
 Lastly, we will show by induction on $j \in \{1,\ldots,m\}$ that $h_j(\{a,b\}) = h_j(\{c,d\})$
 implies $\{a,b\} = \{c,d\}$.
 
 Suppose $j = 1$. If $j = n$, then we must have $g_1(\{a,b\}) = g_1(\{c,d\}) = \{1,2\}$
 and as $g_1 = u_1 \in G$ is a permutation, we find $\{a,b\} = \{c,d\}$.
 Otherwise, if $j < n$, then $h_1(\{a,b\})$ and $h_1(\{c,d\}$
 are both $2$-sets and $0 \in g_1(\{a,b\}) \cap g_1(\{c,d\})$
 by the previous arguments. 
 Hence, $1 \notin g_1(\{a,b\}) \cup g_1(\{c,d\})$.
 We find $1 \in h_1(\{a,b\}\cap h_1(\{c,d\})$
 and the other elements in $g_1(\{a,b\})$ and $g_1(\{c,d\})$
 are mapped to itself by $f$. So, we can deduce that
 we must have $g_1(\{a,b\}) = g_1(\{c,d\})$
 and so, as before, $\{a,b\} = \{c,d\}$.
 Now, suppose $j \in \{2,\ldots,n\}$.
 If $j = n$, then $g_n(\{a,b\}) = g_n(\{c,d\}) = \{1,2\}$, as $h_n(\{a,b\})$ and $h_n(\{c,d\})$ are singleton sets,
 which yields, as $u_n \in G$ is a permutation,
 that $h_{n-1}(\{a,b\}) = h_{n-1}(\{c,d\})$.
 Then, inductively, $\{a,b\} = \{c,d\}$.
 Otherwise, suppose $j \in \{2,\ldots,n-1\}$ and $h_j(\{a,b\}) = h_j(\{c,d\})$.
 Then, as argued before, $0 \in g_j(\{a,b\}) \cap g_j(\{c,d\})$
 and $1 \notin g_j(\{a,b\}) \cup g_j(\{c,d\})$.
 Hence, by Lemma~\ref{lem:transformation_acts_injective},
 we find $g_j(\{a,b\}) = g_j(\{c,d\})$, and then, as $u_j \in G$
 is a permutation $h_{j-1}(\{a,b\}) = h_{j-1}(\{a,b\})$.
 Then, inductively, we can deduce $\{a,b\} = \{c,d\}$.\qed
\end{comment}
\end{proof}
 
 The following lemma allows us to assume
 the transformation of rank $n - 1$ in Theorem~\ref{thm:sync_maximal_equals_primitive} is idempotent.
 
\begin{lemmarep}
\label{lem:idempotent}
 Let $G \le S_n$ be a transitive permutation group and $f : [n] \to [n]$
 be a transformation of rank $n - 1$. Then, there exists
 an idempotent transformation $h : [n] \to [n]$ of rank $n - 1$
 such that $\langle G \cup \{ h \} \rangle \le \langle G \cup \{f\} \rangle$.
\end{lemmarep}
\begin{proof}
 Let $f : [n] \to [n]$ be a transformation of rank $n - 1$. 
 Let $a \in [n]$ be such that $f([n]) = [n] \setminus \{a\}$.
 As $G$ is transitive, 
 there exists $g \in G$ such that $a \notin g(f([n]))$
 and $gf$ permutes $[n] \setminus \{a\}$. Then, some power
 of $gf$ acts as the identity on $[n] \setminus \{a\}$, i.e., is idempotent.\qed
\end{proof}
 
 So, now we have everything together to prove Theorem~\ref{thm:sync_maximal_equals_primitive}.
 
 \begin{proof}[Proof of Theorem~\ref{thm:sync_maximal_equals_primitive}]
  We can assume $n > 2$, as we have not included the assumption
  of transitivity in our definition of primitivity (which is implied for $n > 2$, see~\cite{cameron_1999})
  and so, for $n \le 2$ every subgroup is primitive and also fulfills the second condition vacuously, 
  as then we cannot find two distinct $2$-sets.
  Also, for $n \le 2$, every group is sync-maximal, as is easily seen
  by case analysis.

  So first, let $G = \langle g_1, \ldots, g_k \rangle \le S_n$
  be primitive and suppose $f : [n] \to [n]$
  is a transformation of rank $n - 1$. 
  By Lemma~\ref{lem:idempotent}, there exists 
  an idempotent transformation $f' \in \langle G \cup \{f\}\rangle$.
  By Theorem~\ref{thm:rystsov}
  in $\langle G \cup \{f'\} \rangle$
  we find a constant map. Then, by Proposition~\ref{prop:prim_grp_distinguishes_2_sets},
  for distinct $2$-sets there exists an element in $\langle G \cup \{f'\} \rangle \subseteq \langle G \cup \{f\}\rangle$
  mapping precisely one of both $2$-sets to a singleton set.

  Now, suppose the second condition holds true. 
  First, note that the second condition implies for $n > 2$ and $\{a,b\} \subseteq [n]$
  that
  there must exist $g \in G$
  such that $g(\{a,b\}) \ne \{a,b\}$.
  Assume this is not the case. Then, for $c \notin \{a,b\}$,
  we have 
  \begin{equation}
  \label{eqn:image_acbc}
      \{ g(\{a,c\}), g(\{b,c\}) \} \cap \{ \{a,b\} \} = \emptyset
  \end{equation}
  for every $g \in G$ and, more generally,
  we have $\{ g(\{d,e\}), g(\{d',e'\}) \} \cap \{ \{a,b\} \} = \emptyset$
  for every $g \in G$ and $\{d,e\}, \{d',e'\}$ not equal to $\{a,b\}$.
  Choose $c \in [n]\setminus\{a,b\}$
  and $f' : [n] \to [n]$ idempotent of rank $n - 1$ with $f'(a) = f'(b) = b$.
  %As $f$ is idempotent and has rank $n - 1$, we must have $\{f(a), f(b)\} \subseteq \{a,b\}$
  %and $|f([n]) \cap \{a,b\}| = 1$.
  Then, with Equation~\eqref{eqn:image_acbc} we can deduce
  $\{ h(\{a,c\}), h(\{b,c\}) \} \cap \{ \{a,b\} \} = \emptyset$
  for $h \in \langle G \cup \{f'\} \rangle$.
  So, it is not possible to map one of $\{a,c\}$ and $\{b,c\}$
  to a singleton set.
  But this is excluded by assumption, so there
  must exist an element in $G$ mapping $\{a,b\}$
  to a different $2$-subset.
  
  So, now let $f : [n] \to [n]$
  be an arbitrary transformation of rank $n - 1$ and
  $\mathcal A = (\Sigma, Q, \delta)$
  be the automaton with $\Sigma = \{g_1, \ldots, g_k, f\}$, $Q = [n]$
  and $\delta(i, x) = x(i)$ for $i \in [n]$ and $x \in \Sigma$.
  Then, the second condition precisely says that all non-empty
  $2$-sets are distinguishable in~$\mathcal P_{\mathcal A}$.
  With Theorem~\ref{thm:2-set_distinguishable}, then all non-empty subsets with at least two elements
  are distinguishable in $\mathcal P_{\mathcal A}$.
  So, we only need to show that all non-empty subsets with at least two elements are reachable
  and at least one singleton subset is reachable in $\mathcal A$.
  In fact, we will establish the stronger statement that $\mathcal A$
  is completely reachable. Let $\{a,b\} \subseteq [n]$
  be a $2$-subset. As shown above, we can choose $g \in G$
  with $g(\{a,b\}) \ne \{a,b\}$.
  Then, by assumption, there exists $h \in \langle G \cup \{f\}\rangle$
  such that precisely one of $h(\{a,b\})$ or $(hg)(\{a,b\})$
  is a singleton set.
  By Theorem~\ref{thm:pair_sync_criterion}, as $\{a,b\}$ was arbitrary, $\langle G \cup \{f\} \rangle$
  contains a constant map. As $f : [n] \to [n]$ was arbitrary of rank $n - 1$,
  by Theorem~\ref{thm:rystsov} the group $G$
  is primitive and so, by Theorem~\ref{thm:compl_reach}, the automaton $\mathcal A$
  is completely reachable.

  Finally, suppose the last condition is fulfilled, i.e, $G$ is sync-maximal.
  Then, by Proposition~\ref{prop:sync_max_implies_primitive}, $G$ is primitive.\qed
 \end{proof}
 
 Note that, by Lemma~\ref{lem:idempotent}, in the statements
 of Theorem~\ref{thm:compl_reach} and~\ref{thm:sync_maximal_equals_primitive}
 it is enough if the mentioned conditions hold for idempotent transformations
 of rank $n - 1$ only.

 \begin{toappendix} 
 The following separation result that will be needed
in the proof of Theorem~\ref{thm:char_prim_grp}.

\begin{theorem}~\cite{ArCamSt2017,neumann75b,birch2009}
\label{thm:separating_neumann}
 Let $G \le S_n$ be transitive and $A,B \subseteq [n]$ be
 such that $|A| \cdot |B| < n$. Then, there exists $g \in G$
 such that $g(A) \cap B = \emptyset$.
\end{theorem}
 \end{toappendix}
 
 With a little more work, we can show the following
 statement. By Lemma~\ref{lem:idempotent},
 we can always restrict to idempotent
 transformations for the mentioned characterizations\footnote{In case of Theorem~\ref{thm:compl_reach}, which
was proven in the conference version~\cite{lata1Hoffmann21}, this was communicated to me, for which I am thankful, by an anonymous referee of~\cite{lata1Hoffmann21}}, 
hence every statement entails two statements: one for all
transformations of rank $n - 1$ and one for idempotent transformations of rank $n - 1$ only. 
Both formulations are put into a single statement by putting the word ``idempotent''
in square brackets in Theorem~\ref{thm:char_prim_grp}.
 
 \begin{theoremrep}
\label{thm:char_prim_grp}
 Let $G \le S_n$ be a permutation group and $n \ge 5$. Then the following are equivalent:
  \renewcommand{\labelenumi}{(\arabic{enumi})}
 \begin{enumerate}%[label=(\arabic*)]
 \item \label{prop:primitive} $G$ is primitive;
 
%  \item for any idempotent transformation $f \colon [n] \to [n]$
%  of rank $n-1$, in the transformation semigroup $\langle G \cup \{ f \} \rangle$
%  we find, for each non-empty $S \subseteq [n]$, an element $g \in \langle G \cup \{ f \} \rangle$
%  such that $g([n]) = S$, i.e., the transformation monoid is completely reachable.

   \item \label{prop:compl_reach} for every [idempotent] transformation $f \colon [n] \to [n]$
 of rank $n-1$, in the transformation semigroup $\langle G \cup \{ f \} \rangle$
 we find, for each non-empty $S \subseteq [n]$, an element $g \in \langle G \cup \{ f \} \rangle$
 such that $g([n]) = S$; %, i.e., the transformation monoid is completely reachable.
 %is completely reachable;
  
  \item \label{prop:two_sets} for every [idempotent] transformation $f \colon [n] \to [n]$ 
   of rank $n-1$ and $2$-sets $\{a,b\}, \{c,d\} \subseteq [n]$ with $\{a,b\}\ne \{c,d\}$,
   there exists a transformation in $\langle G \cup \{f\} \rangle$
   mapping precisely one to a singleton, but not the other;

  \item \label{prop:any_sets} for every [idempotent] transformation $f \colon [n] \to [n]$ 
   of rank $n-1$ and two distinct non-empty and non-singleton subsets $S, T \subseteq [n]$,
   there exists a transformation in $\langle G \cup \{f\} \rangle$
   mapping one to a singleton but not the other;
   
  \item \label{prop:diff_size} for every [idempotent] transformation $f \colon [n] \to [n]$ 
   of rank $n-1$ and two distinct non-empty and non-singleton subsets $S, T \subseteq [n]$,
   there exists a transformation in $\langle G \cup \{f\} \rangle$
   mapping both to subsets of \emph{different cardinality};

   \item \label{prop:disjoint_two_sets} for every [idempotent] transformation $f \colon [n] \to [n]$ 
   of rank $n-1$ and two \emph{disjoint} non-empty $2$-sets $\{a,b\}, \{c,d\} \subseteq [n]$, % mit trennung neumann aus vorherigem
   there exists a transformation in $\langle G \cup \{f\} \rangle$
   mapping precisely one to a singleton, but not the other.

% paar kriterium und rystov um wieder auf sync zu kommen
% {a,b}, nimm {a,b},{a,c}
% (i) {a,c} auf singleton
%
% {a,b}, {a,b}^g beliebig mit {a,b}^g != {a,b} (gibt es wegen transitiv) dann eins auf singleton, admit alle paare

 \end{enumerate}
\end{theoremrep}
\begin{proof}  %[Proof of Theorem~\ref{thm:char_prim_grp}]

 By Lemma~\ref{lem:idempotent}, each of the properties 
 holds true if and only if it holds true for idempotent transformations of rank $n - 1$ only.
% \todo{labels setzen!}

 \medskip 
 \noindent\emph{\eqref{prop:primitive} implies \eqref{prop:compl_reach}:}  This is implied by Theorem~\ref{thm:compl_reach}.
 
 \medskip
 \noindent\emph{\eqref{prop:compl_reach} implies \eqref{prop:two_sets}:} This is implied by Theorem~\ref{thm:sync_maximal_equals_primitive}.

 \medskip 
 \noindent\emph{\eqref{prop:two_sets} implies \eqref{prop:any_sets}:} This is implied by Theorem~\ref{thm:2-set_distinguishable}.

 \medskip 
 \noindent\emph{\eqref{prop:any_sets} implies \eqref{prop:diff_size}:} This is clear.

 \medskip 
 \noindent\emph{\eqref{prop:diff_size} implies  \eqref{prop:disjoint_two_sets}:} This is clear.

 \medskip 
 \noindent\emph{\eqref{prop:disjoint_two_sets} implies \eqref{prop:primitive}:} Let $f : [n] \to [n]$ be a transformation of rank $n - 1$.
 Let $\{a,b\} \subseteq [n]$
 be a $2$-set. As $n > 4$, by Theorem~\ref{thm:separating_neumann},
 there exists $g \in G$ such that $g(\{a,b\}) \cap \{a,b\} = \emptyset$.
 By Property~\eqref{prop:disjoint_two_sets} there exists $h \in \langle G \cup \{f\} \rangle$
 such that precisely one of $\{a,b\}, g(\{a,b\})$
 is mapped to a singleton set. Then, either $h$ or $gh$ collapse
 $\{a,b\}$. So, with Theorem~\ref{thm:pair_sync_criterion},
 the transformation monoid $\langle G \cup \{f\} \rangle$
 contains a constant map. Then, with Theorem~\ref{thm:rystsov},
 %and Lemma~\ref{lem:idempotent},
 we find that $G$ is primitive.\qed
\end{proof}

\begin{toappendix}
\begin{remark}
 The assumption $n \ge 5$ is necessary for Theorem~\ref{thm:char_prim_grp} to hold true.
 For example, let $n = 4$ and consider the non-primitive permutation group $G = \langle g \rangle$  
 with $g(0) = 1, g(1) = 2, g(2) = 0$
 and $g(3) = 3$.
 Let $f : [4] \to [4]$ be a transformation of rank $3$
% with $f(x) = f(y)$, $x \ne y$
 and $\{a,b\}, \{c,d\}$ be two disjoint $2$-sets.
 The orbits on the $2$-sets are
 \begin{multline*}
  \{ \{0,3\}^h \mid h \in G \} = \{ \{0,3\}, \{1,3\}, \{2,3\} \}, \\
  \{ \{0,1\}^h \mid h \in G \} = \{ \{0,1\}, \{1,2\}, \{2,0\} \},
 \end{multline*}
 and every $2$-set is in one of both sets.
 So, as $\{a,b\}$ must be in a
 different orbit than $\{c,d\}$ as they are disjoint,
 we can find $fh$ such that precisely one of the two $2$-sets 
 $\{a,b\}$ and $\{c,d\}$
 is mapped to a singleton set. 
 However, note that Property~\eqref{prop:primitive} and
 Property~\eqref{prop:compl_reach} of Theorem~\ref{thm:char_prim_grp}
 are also equivalent for $n \in \{3,4\}$, see~\cite{lata1Hoffmann21}.
\end{remark}
\end{toappendix}

%% file: strongly_sync_max.tex
As the sync-maximal groups turned out to be precisely the primitive permutation groups, can we alter the definition to give us a new class of permutation groups related to the size of the minimal automata for the set of synchronizing words? One first approach might be to demand that, for each transformation of rank $n - k$, if we add this to the group, in the resulting transformation monoid every non-empty set of size at most $n - k$ is reachable and two distinct non-empty and non-singleton subsets $S, T$ of states with $|S|, |T| \in \{ m \mid m \le n - k \} \cup \{n\}$ can be mapped to sets of different cardinality.
However, it is easy to show that every such group is $k$-reachable
as introduced in~\cite{lata1Hoffmann21}.
So, also with the results from~\cite{lata1Hoffmann21},
for $6 \le k \le n - 6$ this condition
is only fulfilled by the symmetric or the alternating groups.

So, in the following definition, we only demand the distinguishability conditions, but not the reachability condition. Note that in the characterizations of primitive groups
given above, both conditions -- either distinguishability or reachability -- are equivalent if we add a transformation of rank $n - 1$.

\begin{definition}\todo{rank definieren.}
 A permutation group $G \le S_n$
 is called \emph{strongly sync-maximal}
 if for each transformation $f : [n] \to [n]$ of rank $r$
 with $2 \le r \le n-1$
 in $\langle G \cup \{f\} \rangle$
 all $2$-subsets are distinguishable.
\end{definition}

\begin{propositionrep}
 Every $4$-transitive group
 is strongly sync-maximal
 and every strongly sync-maximal group
 is primitive.
\end{propositionrep}
\begin{proof}
 Let $G \le S_n$ be $4$-transitive
 and $f : [n] \to [n]$
 be a transformation of rank $r$
 with $2 \le r \le n - 1$.
 Then, there exist $\{a,b\}, \{c,d\} \subseteq [n]$ with $|f(\{a,b\})| = 2$
 and $|f(\{c,d\})| = 1$.
 Now, suppose we have arbitrary distinct $2$-subsets $\{a_1, a_2\}, \{a_3,a_4\}$ of $[n]$.
 By assumption, there exists $g \in G$
 such that $g(\{a_1, a_2\}) = \{a,b\}$
 and $g(\{a_3, a_4\}) = \{c,d\}$.
 Then, $fg$ maps precisely one of these two $2$-subsets
 to a singleton but not the other.

 Now, suppose $G \le S_n$ is strongly sync-maximal.
 By Theorem~\ref{thm:sync_maximal_equals_primitive}, by only considering
 transformations of rank $n - 1$,
 we find that $G$ is primitive.\qed
\end{proof}

Whether the strongly sync-maximal groups
are properly contained between
the above mentioned groups is an open problem. If so, the precise relation
to the synchronizing groups and other classes of groups is an open problem
and remains for future work.